\documentclass[letterpaper, 10 pt, conference]{ieeeconf}  

\IEEEoverridecommandlockouts                              

\usepackage{cite}
\usepackage{amsmath,amssymb,amsfonts}
\usepackage{algorithmic}
\usepackage{graphicx}
\usepackage{textcomp}
\usepackage[hyphens]{url}
\usepackage{bbold}
\usepackage{subfigure}

\newtheorem{theorem}{Theorem}[section]

\newtheorem{proposition}[theorem]{Proposition}

\newtheorem{conj}[theorem]{Conjecture}

\newcounter{assumption}[section]

\newcommand{\indep}{\perp \!\!\! \perp}

\newcommand{\meig}{\rho_{\max}}
\newcommand{\tp}{\mathrm{T}}

\renewcommand{\t}[0]{\s{t}}
\newcounter{pfxc}[section]

\newcommand{\rvec}[1]{{\boldsymbol{\mathbf{\MakeLowercase{#1}}}}}

\newcommand{\rs}[1]{{\boldsymbol{\mathbf{\mathsf{{#1}}}}}}
\newcommand{\s}[1]{\mathsf{#1}}

\begin{document}
\title{Online variable-length source coding for minimum bitrate LQG control}
\author{Travis C. Cuvelier, Takashi Tanaka, and Robert W. Heath, Jr.
\thanks{The authors gratefully acknowledge that this work was supported in part by the US Air Force Research Laboratory under grant FA8651-22-1-0013.}
\thanks{T. Cuvelier and T. Tanaka are with the Oden Institute for Computational Engineering and Sciences, The University of Texas at Austin,
TX, 78712 USA e-mails: tcuvelier@utexas.edu, ttanaka@utexas.edu.}
\thanks{R. Heath is with the Department
of Electrical and Computer Engineering, North Carolina State University, Raleigh, 
NC, 27606 USA e-mail: rwheath2@ncsu.edu.}}

\maketitle

\begin{abstract} We propose an adaptive coding approach to achieve linear-quadratic-Gaussian (LQG) control with near-minimum bitrate prefix-free feedback. Our approach combines a recent analysis of a quantizer design for minimum rate LQG control with work on universal lossless source coding for sources on countable alphabets. In the aforementioned quantizer design, it was  established that the quantizer outputs are an asymptotically stationary, ergodic process. To enable LQG control with provably near-minimum bitrate, the quantizer outputs must be encoded into binary codewords efficiently. This is possible given knowledge of the probability distributions of the quantizer outputs, or of their limiting distribution. Obtaining such knowledge is challenging; the distributions do not readily admit closed form descriptions. This motivates the application of universal source coding. Our main theoretical contribution in this work is a proof that (after an invertible transformation), the quantizer outputs are random variables that fall within an exponential or power-law envelope class (depending on the plant dimension). Using ideas from universal coding on envelope classes, we develop a practical, zero-delay version of these algorithms that operates with fixed precision arithmetic. We evaluate the performance of this algorithm numerically, and demonstrate competitive results with respect to fundamental tradeoffs between bitrate and LQG control performance.  
\end{abstract}
\section{Introduction}\label{sec:introduction} 
In this work, we propose an algorithmic framework to achieve any feasible discrete-time LQG control performance with near minimum bitrate variable-length feedback. The motivation behind our problem formulation is feedback control where sensors and actuators are not co-located, and measurements must be fed back from sensor to controller over digital communications. In such settings, the feedback bitrate is an effective surrogate for the true \textit{cost} of communication. For various communication architectures, a communication link's bitrate corresponds more-or-less directly with the physical layer resources utilized by said link. 

While there are several results in the literature that pertain to LQG control with minimum (variable-length) feedback bitrate \cite{silvaFirst,tanakaISIT,kostinaTradeoff,ourJSAIT}, there are several barriers to applying these results in practice.  This work aims to bridge this gap. We propose practical algorithms that combine the approach to minimum rate LQG control from \cite{ourJSAIT} with approaches to universal coding on countably infinite alphabets from \cite{frenchOG,frenchExp,frenchSeq,frenchSeq2}. A numerically precise implementation then follows from a classical approach to arithmetic coding with fixed-precision arithmetic via \cite{wittenArithmetic} as well as Elias's ``penultimate" (omega) encoding for integers \cite{eliasUniversal}. We evaluate the proposed algorithms numerically on a model for an inverted pendulum system \cite{ssmodelinvertedpendulum}. Our approach is shown to achieve good performance with respect to the fundamental directed information bitrate lower bound (cf. \cite{ourConverseLetter}).  
\subsection{Literature Review} 
Data compression architectures for LQG control with near-minimum bitrate variable-length feedback were proposed in  \cite{silvaFirst,tanakaISIT,kostinaTradeoff,ourJSAIT}. These approaches propose to quantize sensor measurements with prescribed quantizer designs, and assume that the quantizers' outputs are losslessly encoded into binary packets. If this can be accomplished, the architectures in \cite{silvaFirst,tanakaISIT,kostinaTradeoff,ourJSAIT} provably attain bitrates close to established directed information (DI) lower bounds (cf. \cite{ourConverseLetter,kostinaTradeoff}). Accomplishing this lossless encoding, however, is quite difficult in practice. In \cite{silvaFirst,tanakaISIT,kostinaTradeoff,ourJSAIT} the quantizer outputs have countable support. The quantizer outputs are a nonstationary random process; in  \cite{silvaFirst,tanakaISIT,kostinaTradeoff} the lossless encoding is assumed to be adapted precisely to the output distribution. While  \cite{ourJSAIT} demonstrated that it suffices to use a time-invariant encoding adapted to the limiting distribution of the quantizer's outputs; it remains a challenging problem to estimate, and design a code for, this limiting distribution. In contrast to some recent work on LQG control with fixed-rate feedback (cf. e.g. \cite{fixedLenCoding}), there are few results on practical algorithms that have been shown to exhibit competitive performance with respect to established fundamental tradeoffs between bitrate and control performance (cf. \cite{ourConverseLetter,kostinaTradeoff}).

Despite recent positive results on LQG control with fixed-length feedback, the less restrictive setting of variable-length feedback offers some advantages. It is well established that a linear plant driven by unbounded process noise cannot be stabilized in the mean square sense with feedback that undergoes time-invariant, memoryless, fixed-length quantization \cite{nairevans}. While  several works on control with \textit{adaptive} fixed-length quantization guarantee stability, e.g.  \cite{yukselStabilization}, \cite{r3_add3}, they are not guaranteed to operate near the fundamental rate-cost tradeoff for LQG problems. In \cite{fixedLenCoding}, fixed-length quantizers that optimize LQG control performance were designed via the Lloyd-Max algorithm.  While numerical experiments exhibited competitive performance was achieved with respect to fundamental tradeoffs, the approach in \cite{fixedLenCoding} requires the Lloyd-Max quantizer to be ``re-designed" at every timestep. Stability is not guaranteed. In contrast, it is relatively easy to guarantee a fixed LQG performance with variable-length feedback; rare events that would effectively saturate fixed-length quantizers are encoded into long bitstrings. 

In this work, we combine the quantizer design and analysis from \cite{ourJSAIT} with relatively recent work on universal lossless source coding on countable alphabets. Broadly speaking, ``universal coding" refers to techniques for the design of lossless codes for sources with unknown distributions. The quantizer outputs from \cite{ourJSAIT} are an asymptotically stationary random process with countable support. Their limiting distribution evades closed form description. While there is a wealth of literature devoted to universal coding on alphabets with finite support, there are comparatively fewer results for alphabets with countably infinite support \cite{frenchOG,frenchExp,frenchSeq,frenchSeq2}. In particular,  \cite{frenchOG} introduced the notion of universal coding on  \textit{envelope classes}. An envelope class is a set of probability mass functions on $\mathbb{N}_{+}$ (interpreted as potential source distributions) that are upper bounded by a monotonically decreasing, $\ell^{1}$, \textit{envelope functions}. Relevant envelope functions include power-laws and exponentials. Universal block coding techniques for IID stationary sources with marginals that fall into envelope classes are introduced in \cite{frenchOG,frenchExp,frenchSeq,frenchSeq2}. The coding algorithms in 
\cite{frenchOG,frenchExp,frenchSeq,frenchSeq2} take a censoring approach; symbols that are below a (time-varying) \text{cutoff} are encoded using arithmetic coding, and large-valued symbol realizations that exceed the cutoff are treated differently. If the symbol value exceeds the cutoff, an escape symbol is encoded arithmetically and the observed symbol value (or the difference between the symbol value and the cuttoff) is encoded via one of Elias's universal prefix codes for integers \cite{eliasUniversal}. The probability model used for arithmetic coding is empirically adapted to reflect the source. 

While the algorithms in \cite{frenchOG,frenchExp,frenchSeq,frenchSeq2}  permit sequential encoding and decoding, they are fundamentally block codes. The use of arithmetic coding can require that bits from  $\s{n}$ encoded symbols be received before the first symbol can be  decoded. Given a block of $\s{n}$ symbols to encode, the algorithms in \cite{frenchOG} and \cite{frenchExp} can be shown to achieve a redundancy (e.g. a codeword length in excess of the entropy of the source) that has  $o(\s{n})$. If, instead of using arithmetic coding, the symbols below the cuttoff are encoded with a Shannon-Fano-Elias (SFE) code (\cite[Section IV.A.1]{ourJSAIT}) using the empirical source model, the encoding schemes in \cite{frenchOG,frenchExp,frenchSeq,frenchSeq2} can be made to incur zero encoding and decoding delay at the expense of at most $1$ bit per encoded symbol. We discuss \cite{frenchOG,frenchExp,frenchSeq,frenchSeq2} in the context of our present work in the next subsection. 

\subsection{Our contributions} This work makes one theoretical and one practical contribution. On the theoretical side, we prove that, after a transformation, the limiting distribution of the quantizer outputs in \cite{ourJSAIT}'s source coding architecture fall into an exponential or power-law envelope class (depending on the plant dimension). The approaches in \cite{frenchOG,frenchExp,frenchSeq,frenchSeq2} apply to \textit{stationary} sources. Given the asymptotic stationarity and the Kullback–Leibler (KL) sense convergence of the quantizer's outputs to their limiting distribution in \cite{ourJSAIT}, we conjecture that using a ``zero-delay version" of an algorithm from \cite{frenchOG} or \cite{frenchExp} to encode the quantizer's outputs will achieve the same bitrates predicted by \cite[Theorem IV.3 (ii)]{ourJSAIT}. 

The coding schemes in \cite{frenchOG,frenchExp,frenchSeq,frenchSeq2} are not strictly practical, as they require arithmetic-style prefix coding to be performed on larger and larger alphabets. This requires that the (arithmetic) precision with which calculations are performed at the encoder and decoder to likewise expand \cite{wittenArithmetic}. We devise a simplified, zero-delay version of the censoring codes using fixed-precision arithmetic. We demonstrate competitive performance with respect to the lower bounds from \cite{ourConverseLetter} and the upper bounds in \cite{ourJSAIT}. 

\subsection{Notation} We use lower case serif letters for vectors $v$, upper case serif letters for matrices $V$, and sans serif lower case letters $\s{v}$ for scalars. Random variables are in boldface, e.g. a random vector is denoted $\rvec{v}$. For sequences $\{{x}_{0},{x}_{1},\dots\}$ we let ${x}^{\s{n}}$ denote $\{ {x}_{0}, 
\dots ,{x}_{\s{n}} \}$. If $x\in\mathbb{R}^{\s{m}}$, for $\s{i}\in \{0,\dots,\s{m}-1\}$ we let $[x]_{\s{i}}$ denote the element of $x$ in position $\s{i} $ (we use zero-based indexing). If $\rs{x}$ is a discrete random variable, we let $\mathbb{P}_{\rs{x}}(\s{x})=\mathbb{P}_{\rs{x}}[\rs{x}=\s{x}]$.  The Shannon entropy of the discrete random variable $\rvec{q}$, in bits, is denoted $H(\rvec{q})$. Likewise, the Kullback–Leibler (KL) divergence  (relative entropy) between $\mathbb{P}_{\rvec{a}}$ and $\mathbb{P}_{\rvec{b}}$, in bits, is denoted $D_{\mathrm{KL}}(\rvec{a}||\rvec{b})$.   We write $\rvec{x}\overset{\mathrm{D}}{=}\rvec{y}$ if $\rvec{x}$ and $\rvec{y}$ are identically distributed. Denote the max singular value of the matrix $X$ via $\lVert X \rVert_{2}$ and the spectral radius $\meig(X)$. We use standard notation for vector norms and information theoretic measures. The set of finite length binary strings is denoted $\{0,1\}^{*}$, and the the length of $a\in \{0,1\}^{*}$, in bits, is $\ell(a)$.

\section{System model and problem formulation}\label{sec:sys}
In this section, we define the system model, and give high level descriptions of DI bitrate lower bounds and the compression architecture from \cite{ourJSAIT}. 

We consider the system model depicted in Figure \ref{fig:ditharch}, which depicts a discrete-time, Gauss-Markov plant that is controlled using feedback that is encoded into variable-length binary strings (or packets). 
Denote the state vector $\rvec{x}_{\s{t}}\in \mathbb{R}^{\s{m}}$, the control input $\rvec{u}_{\s{t}}\in\mathbb{R}^{\s{u}}$, and let $\rvec{w}_{\s{t}}\sim\mathcal{N}({0}_{\s{m}},{W})$ denote processes noise assumed to be IID over time. We assume ${W}\succ{0}_{\s{m}\times \s{m}}$, i.e., the process noise covariance is full rank. We assume assume that the initial state has $\rvec{x}_{0}\sim\mathcal{N}({0},{X}_0)$ for some $X_0\succeq 0$.  For some ${A}\in\mathbb{R}^{\s{m}\times \s{m}}$ and ${B}\in\mathbb{R}^{\s{m} \times \s{u}}$, the plant dynamics for $\s{t} \ge 0$ are $\rvec{x}_{\s{t}+1} = {A}\rvec{x}_{\s{t}}+{B}\rvec{u}_{\s{t}}+\rvec{w}_{\s{t}}$. To ensure finite control cost is attainable, we assume $({A},{B})$ are stabilizable. The plant is fully observable to a sensor/encoder block. At each discrete timestep $\t$, the sensor/encoder makes measurements of the plant, and encodes its measurements into the codeword $\rvec{a}_{\t}\in \{0,1\}^{*}$ via an arbitrary causal encoding policy. The encoder then transmits its codeword  $\rvec{a}_{\t}$ over the feedback channel to a combined decoder/controller. The decoder/controller uses the packets it receives to design the control input, $\rvec{u}_{\t}$, again via an arbitrary causal policy. We assume that the feedback channel is reliable, e.g. that the decoder receives the packet $\rvec{a}_{\t}$ without errors and without delay. We assume that the sensor/encoder and decoder/controller have access to common \textit{dither sequence}. The dither is a sequence of exogenous random vectors $\rvec{\delta}_{\t}\in\mathbb{R}^{\s{m}}$ that are revealed causally to both encoder and decoder. We assume that $[\rvec{\delta}_{\t}]_{\s{i}}\sim\text{Uniform}([-\frac{\Delta}{2},\frac{\Delta}{2}])$ IID over $\s{i}$ and $\t$. The dither can be used to select codewords/control inputs at the respective blocks. In real-world systems, \textit{shared randomness} of this nature can be approximately accomplished via synchronizing pseudorandom number generators at the encoder and decoder.  The encoder/sensor policy in Fig. \ref{fig:ditharch} is a sequence of causally conditioned Borel measurable kernels denoted $ \mathbb{P}_{\mathrm{E}}[\rvec{a}_{0}^{\infty}|| \rvec{\delta}_{0}^{\infty},\rvec{x}_{0}^{\infty}] = \left\{ \mathbb{P}_{\mathrm{E},\s{t}}=\mathbb{P}_{\rvec{a}_{\s{t}}|\rvec{a}_{0}^{\s{t}-1},\rvec{\delta}_{0}^{\s{t}},\rvec{x}_{0}^{\s{t}}}\right\}_{\s{t}}$. The corresponding decoder/controller policy is given by 
    $\mathbb{P}_{\mathrm{C}}[\rvec{u}_{0}^{\infty}|| \rvec{a}_{0}^{\infty},\rvec{\delta}_{0}^{\infty}] = \left\{ \mathbb{P}_{\mathrm{C},\s{t}}=\mathbb{P}_{\rvec{u}_{\s{t}}|\rvec{a}_{0}^{\s{t}},\rvec{\delta}_{0}^{\s{t}},\rvec{u}_{0}^{\s{t}-1}}\right\}_{\s{t}}$.
Conditional independence assumptions between system variables are imposed via factorizations of the system's transition kernels; this is discussed in Fig. \ref{fig:ditharch}.
\begin{figure}
	\centering
	\includegraphics[scale = .17]{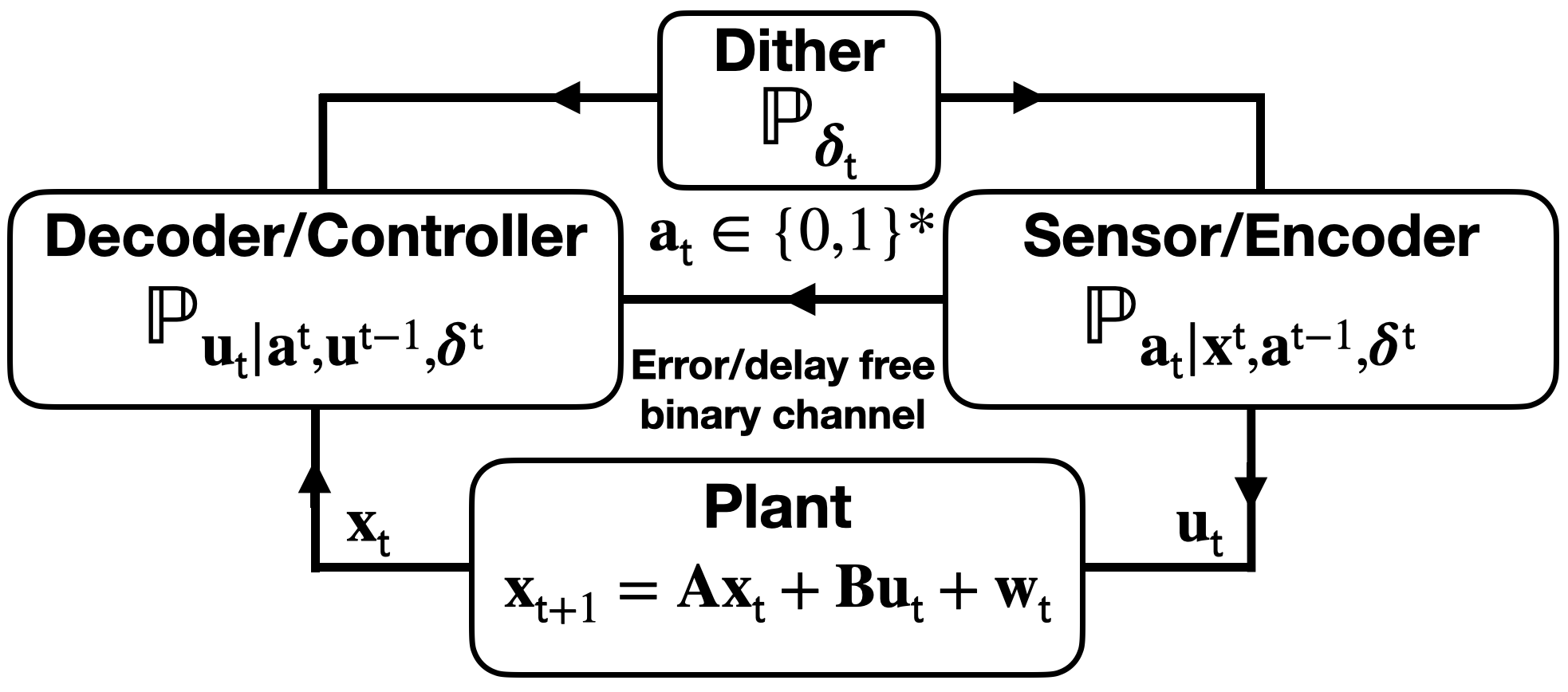}
	\caption{The codeword $\rvec{a}_{\s{t}}$ can be generated randomly conditioned on the information known at the encoder at time $\s{t}$. Likewise, after receiving $\rvec{a}_{\s{t}}$, decoder can randomly generate $\rvec{u}_{\s{t}}$ given $\rvec{a}_{\s{t}}$ and its previous knowledge. The dither $\rvec{\delta}_{\s{t}}$ is generated independently of all past system variables. Under the dynamics, $\rvec{x}_{0}^{\s{t}}$ is a deterministic function of $\rvec{x}_{0}$, $\rvec{a}_{0}^{\s{t}-1}$, $\rvec{u}_{0}^{\s{t}-1}$, and $\rvec{w}_{0}^{\s{t}-1}$. For $\t \ge 0$ we assume the factorizations $\mathbb{P}_{\rvec{a}_{\s{t}+1},\rvec{u}_{\s{t}+1}|\rvec{a}_{0}^{\s{t}},\rvec{\delta}_{0}^{\s{t}},\rvec{u}_{0}^{\s{t}},\rvec{w}_{0}^{\s{t}},\rvec{x}_{0}} = \mathbb{P}_{\mathrm{E},\s{t}+1}\mathbb{P}_{\mathrm{C},\s{t}+1}$  and  $\mathbb{P}_{\rvec{a}_{\s{t}+1},\rvec{\delta}_{\t+1},\rvec{u}_{\s{t}+1},\rvec{w}_{\t+1}|\rvec{a}_{0}^{\s{t}},\rvec{\delta}_{0}^{\s{t}},\rvec{u}_{0}^{\s{t}},\rvec{w}_{0}^{\s{t}},\rvec{x}_{0}} = \mathbb{P}_{\mathrm{E},\s{t}+1}\mathbb{P}_{\mathrm{C},\s{t}+1}\mathbb{P}_{\rvec{\delta}}\mathbb{P}_{\rvec{w}}$, where $\mathbb{P}_{\rvec{w}}$ and $\mathbb{P}_{\rvec{\delta}}$ are the marginal distributions of the process noise and dither.  We assume that initially  $\mathbb{P}_{\rvec{a}_{0},\rvec{\delta}_{0},\rvec{u}_{0},\rvec{w}_{0}|\rvec{x}_{0}}  = \mathbb{P}_{\mathrm{E},0}\mathbb{P}_{\mathrm{C},0}\mathbb{P}_{\rvec{\delta}}\mathbb{P}_{\rvec{w}}$. These factorizations encode assumed conditional independence assumptions between system variables. The lower bounds in \cite{ourConverseLetter} follow as a consequence.
}\label{fig:ditharch}
\vspace{-.5cm}
\end{figure}

We require that the the codewords produced by the encoder be \textit{prefix-free} in the following sense: for all $\t$ and any realizations ($\rvec{a}_{0}^{\s{t}-1}={a}_{0}^{\s{t}-1},\rvec{\delta}_{0}^{\s{t}}={\delta}_{0}^{\s{t}},\rvec{u}_{0}^{\s{t}-1}={u}_{0}^{\s{t}-1}$), for all distinct $a_1,a_2\in\{0,1\}^*$ with $\mathbb{P}_{\rvec{a}_{\s{t}}|\rvec{a}_{0}^{\s{t}-1},\rvec{\delta}_{0}^{\s{t}},\rvec{u}_{0}^{\s{t}-1}}[\rvec{a}_{\s{t}}=a_1|\rvec{a}_{0}^{\s{t}-1}={a}_{0}^{\s{t}-1},\rvec{\delta}_{0}^{\s{t}}={\delta}_{0}^{\s{t}},\rvec{u}_{0}^{\s{t}-1}={u}_{0}^{\s{t}-1}]>0$ and $\mathbb{P}_{\rvec{a}_{\s{t}}|\rvec{a}_{0}^{\s{t}-1},\rvec{\delta}_{0}^{\s{t}},\rvec{u}_{0}^{\s{t}-1}}[\rvec{a}_{\s{t}}=a_2|\rvec{a}_{0}^{\s{t}-1}={a}_{0}^{\s{t}-1},\rvec{\delta}_{0}^{\s{t}}={\delta}_{0}^{\s{t}},\rvec{u}_{0}^{\s{t}-1}={u}_{0}^{\s{t}-1}]>0$, $a_1$ is not a prefix of $a_2$.
This assumption ensures that the decoder can instantaneously decode the codeword at time $\t$. However, it permits the encoder to use different prefix-free codebooks for different realizations of $(\rvec{a}_{0}^{\s{t}-1},\rvec{\delta}_{0}^{\s{t}},\rvec{u}_{0}^{\s{t}-1})$  (cf. \cite[Prefix Constraint 1]{ourJSAIT} \cite[Assumption 2]{ourConverseLetter}). 

 We are interested in the tradeoff between LQG control performance and communication cost, quantified by the time-average expected codeword length. Mathematically, this is described by the optimization over policies $\mathbb{P}_\mathrm{E}, \mathbb{P}_{\mathrm{C}}$ that conform to the prefix constraint:
\begin{align}
\mathcal{L}(\gamma) = \left\{ \begin{aligned}
& \underset{\mathbb{P}_\mathrm{E}, \mathbb{P}_{\mathrm{C}}}{\inf} \text{ }\underset{\s{T}\rightarrow\infty}{\lim\sup} \text{ }\frac{1}{T}\sum\nolimits_{\s{t}=0}^{\s{T}-1}\mathbb{E}[\ell(\rvec{a}_{\s{t}})] \\ &\text{s.t. }  \underset{\s{T}\rightarrow\infty}{\lim\sup}  \frac{1}{\s{T}}\sum\nolimits_{\s{t}=0}^{\s{T}-1}\mathbb{E}\left[\lVert \rvec{x}_{\s{t}+1} \rVert_{{Q}}^{2} +\lVert \rvec{u}_{\s{t}} \rVert_{{\Phi}}^{2}\right] \le \gamma,\nonumber
\end{aligned} \right.
\end{align} where ${Q}\succeq {0}$, ${\Phi}\succ {0}$, and $\gamma$ is the maximum tolerable LQG cost. Let $S$ be a stabilizing solution to the discrete algebraic Riccati equation (DARE) $A^{\mathrm{T}}SA-S-A^{\mathrm{T}}SB(B^{\mathrm{T}}SB+\Phi)^{-1}B^{\mathrm{T}}SA+Q = 0$, $K=-(B^{\mathrm{T}}SB+\Phi)^{-1}B^{\mathrm{T}}SA$, and $\Theta = K^{\mathrm{T}}(B^{\mathrm{T}}SB+\Phi)K$. Consider the optimization 
\begin{align}\label{eq:threestageRDF}
    \mathcal{R}(\gamma) = \left\{ \begin{aligned}
& \underset{\substack{P,\Pi, \in\mathbb{R}^{m\times m}\\P,\Pi\succeq 0} }{\inf} \frac{1}{2}(-\log_{2}{\det{\Pi}}+\log_{2}{\det{W}} )\\ &\text{ }\text{s.t. }  \mathrm{Tr}(\Theta P)+\mathrm{Tr}(WS)\le \gamma\text{,  }\\&\text{ } P\preceq APA^\mathrm{T}+W\text{, } \\&\text{ }\text{ }\text{ }\begin{bmatrix}P-\Pi & PA^{\mathrm{T}} \\ AP & APA^{\mathrm{T}}+W \end{bmatrix}\succeq0 
\end{aligned}\right. 
\end{align} $\mathcal{R}(\gamma)$ is the minimum DI that must be incurred by any policy that achieves an LQG cost less than or equal to $\gamma$. This DI is a fundamental lower bound on the bitrate $\mathcal{L}(\gamma)$;  we have that
$\mathcal{R}(\gamma)\le \mathcal{L}(\gamma)$ \cite[Theorem 1]{ourConverseLetter} \cite{SDP_DI}. The minimizing $P$ from (\ref{eq:threestageRDF}), denoted $\hat{P}$, can be used to construct encoder and decoder policies that nearly achieve the lower bound $\mathcal{R}(\gamma)$\cite{ourJSAIT}. This approach, depicted in Fig. \ref{fig:overviewachiev}, is presently summarized. 
\begin{figure}
	\centering
	\includegraphics[scale = .19]{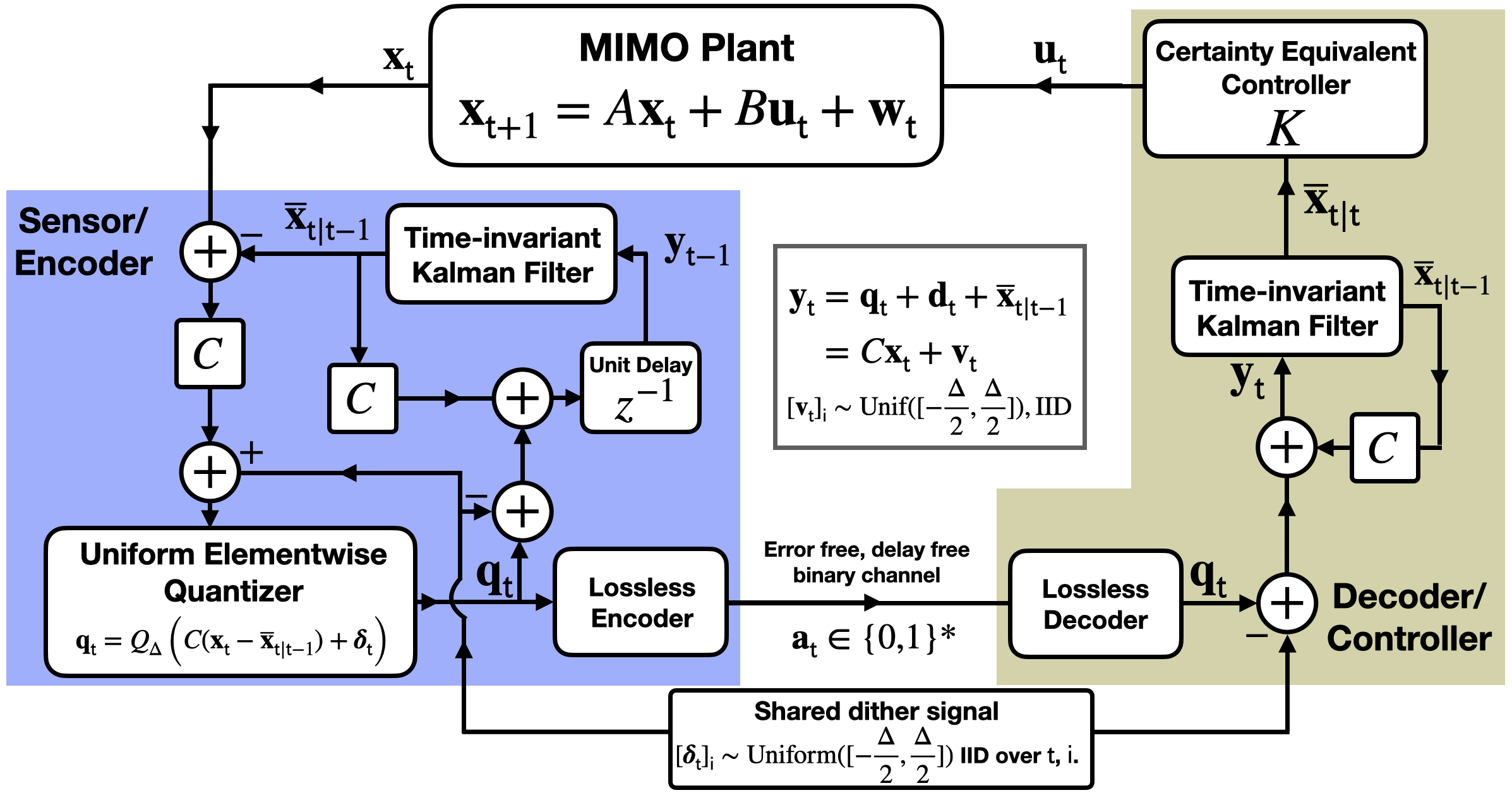}
	\caption{The quantizer and controller design from \cite{ourJSAIT}. }\label{fig:overviewachiev}
\end{figure}

Given $\hat{P}$, define $\hat{P}_{+}=A\hat{P}A^{\mathrm{T}}+W$. Assume that $\Delta =1$, and that $C$ is chosen such that $\frac{1}{12}(\hat{P}^{-1}-\hat{P}_{+}^{-1}) = C^{\tp}C$. The encoder and decoder operate synchronized time-invariant Kalman filters  (KF) that are updated with measurements constructed from the encoded codewords. Let  $\overline{\rvec{x}}_{\t|\t-1}$ denote these filters' a priori estimate, and $\overline{\rvec{x}}_{\t|\t}$ their posterior estimate. Assume initially that $\overline{\rvec{x}}_{1|0} = 0$. By construction, both the encoder and decoder will have computed the same a priori estimate $\overline{\rvec{x}}_{\t|\t-1}$ by the beginning of timestep $\t$. Define the Kalman filter's a priori error via $\rvec{e}_{\t} = \rvec{x}_{\t} - \overline{\rvec{x}}_{\t|\t-1}$. Let $Q_{\Delta}:\mathbb{R}^{\s{m}}\rightarrow\mathbb{Z}^{\s{m}}$ be a function that rounds its input elementwise to the nearest integer, e.g.  $[Q_{\Delta}(x)]_{\s{i}} = k$ if $[x]_{\s{i}}\in [k-\frac{1}{2},k+\frac{1}{2})$. At every time $\t$ the encoder produces a \textit{dithered quantization} of the Kalman innovation via $\rvec{q}_{\t} = Q_{\Delta}(C\rvec{e}_{\t}+\rvec{\delta}_{\t})$. The quantization $\rvec{q}_{\t}$ is a discrete random variable with countably infinite support on $\mathbb{Z}^{\s{m}}$. The quantization $\rvec{q}_{\t}$ is encoded into the codeword $\rvec{a}_{\t}$ losslessly. Since the decoder receives $\rvec{a}_{\t}$ without error, it can reconstruct $\rvec{q}_{\t}$ exactly. The decoder then computes a centered, reconstructed measurement $\rvec{y}_{\t} = \rvec{q}_{\t}- \rvec{\delta}_{\t} + C\overline{\rvec{x}}_{\t|\t-1}$. It is shown in \cite[Eqn. (19)]{ourJSAIT} that $ \rvec{y}_{\t} = C\rvec{x}_{\t} +\rvec{v}_{\t}$, where $[\rvec{v}_{\t}]_{\s{i}}\sim \text{Uniform}([-\frac{1}{2},\frac{1}{2}])$ IID over $\s{i}$ and 
$\rvec{v}_{\t}\indep\rvec{x}_{0}^{\t}$. Since the encoder knows the dither signal, it can also compute $\rvec{y}_{\t}$. Both the encoder and decoder update their KFs using the time-invariant gain $J= \hat{P}_{+}{C}^{\mathrm{T}}({C}\hat{P}_{+}{C}^{\mathrm{T}}+I_{\s{m}})^{-1}$, computing $\overline{\rvec{x}}_{\t|\t} = \rvec{\overline{x}}_{\s{t}|\s{t}-1}+J(\rvec{y}_{\s{t}}-C\rvec{\overline{x}}_{\s{t}|\s{t}-1})$. The decoder then applies certainty equivalent control, selecting $\rvec{u}_{\t} = K\overline{\rvec{x}}_{\t|\t}$. The encoder can likewise compute $\rvec{u}_{\t}$, and both the encoder and decoder KFs compute the predict update via $\overline{\rvec{x}}_{\t+1|\t}=A\overline{\rvec{x}}_{\t|\t}+B\rvec{u}_{\t}$. The next proposition highlights some relevant results from \cite{ourJSAIT} that we use later.
\begin{proposition}\label{prop:prev}
Let $L=AJ$ and $R=A-LC$. Let  $\{\rvec{\nu}_{\s{t}}\}$ denote an IID sequence of random variables uniformly distributed on $[-\s{1}/2,\s{1}/2]^{\s{m}}$, let $\{\rvec{\omega}_{\s{t}}\}$ be IID with $\rvec{\omega}_{\s{t}}\sim\mathcal{N}(0_{\s{m}},W)$, and let $\rvec{\chi}\sim\mathcal{N}(0_{\s{m}},X_{0})$.  Let $\{\rvec{\omega}_{\s{t}}\}$, $\{\rvec{\nu}_{\s{t}}\}$, and $\rvec{\chi}$ be mutually independent.
We have that $R$ is (discrete-time) globally asymptotically stable, e.g. $\meig(R)<1$, that
\begin{align}\label{eq:identicallydistributed}
    \rvec{e}_{\t} \overset{\mathrm{D}}{=} R^{\s{t}}\rvec{\chi}+\sum_{\s{i}=0}^{\s{t}-1}R^{\s{i}}(\rvec{\omega}_{\s{i}}-L\rvec{\nu}_{\s{i}}),
\end{align} and that there exists a random variable $\rvec{e}\sim \mathbb{P}_{\rvec{e}}$ such that 
$(\rvec{e}_{\t},\rvec{\delta}_{\t})$ converge in total variation (and thus weakly) to $(\rvec{e},\rvec{\delta})$ where $\rvec{\delta}\indep \rvec{e}$ and $\rvec{\delta} \sim \text{Uniform}([-\frac{1}{2},\frac{1}{2}]^{\s{m}})$. 
Furthermore $\rvec{e}\in\mathbb{R}^{\s{m}}$ has a strictly positive probability density function and the process  $(\rvec{e}_{\t},\rvec{\delta}_{\t})$ is ergodic in the sense that that if $f:\mathbb{R}^{\s{m}}\times [-1/2,1/2]^{\s{m}}\rightarrow \mathbb{R}$ has 
$\mathbb{E}[|f(\rvec{e},\rvec{\delta})|]<\infty$ then \begin{align}
    \lim_{\s{T}\rightarrow\infty}\frac{1}{\s{T}}\sum_{\s{i}=0}^{\s{T}-1}f(\rvec{e}_{\s{i}},\rvec{\delta}_{\s{i}}) = \mathbb{E}[f(\rvec{e},\rvec{\delta})]\text{, almost surely.}\label{eq:ergodic}
\end{align}
Defining $\rvec{q} = Q_{\Delta}(C\rvec{e}+\rvec{\delta})$, we have that the $\rvec{q}_{\t}$ converge in total variation (and thus weakly) to $\rvec{q}$ and furthermore that $\lim_{\t\rightarrow\infty }D_{\mathrm{KL}}(\rvec{q}_{\t}||\rvec{q}) = 0$.  Finally, we have that both ${\lim\sup}_{\t\rightarrow\infty}H(\rvec{q}_{\t}) \le  R(\gamma) + \s{b}$ and $H(\rvec{q}) \le  R(\gamma) + \s{b}$ where $\s{b} =\s{m}\left(1+\frac{1}{2}\log_{2}\left(\frac{2\pi e}{12}\right)\right)$. 
So long as the $\rvec{q}_{\t}$ are losslessly conveyed to the decoder, LQG control cost satisfies
\begin{align}\label{eq:lqgbound}
    \underset{\s{T}\rightarrow\infty}{\lim\sup}  \frac{1}{\s{T}}\sum\nolimits_{t=0}^{T-1}\mathbb{E}[\lVert \rvec{x}_{t+1} \rVert_{{Q}}^{2} +\lVert \rvec{u}_{t} \rVert_{{\Phi}}^{2}] \le \gamma. 
\end{align}
\end{proposition}
\begin{proof}
That $\meig(R)<1$ is addressed in \cite[Section IV.C]{ourJSAIT}, (\ref{eq:identicallydistributed}) is addressed in the second paragraph of the proof of \cite[Lemma 4.8]{ourJSAIT}. The existence and ergodic properties of the limiting measures $\mathbb{P}_{\rvec{e},\rvec{\delta}}$ and $\mathbb{P}_{\rvec{q}}$ are main results of \cite{ourJSAIT}. These results are discussed in \cite[Theorem IV.3 and Section IV.C, in particular Lemma IV.5]{ourJSAIT}.  The KL-sense convergence of $\rvec{q}_{\t}$ to $\rvec{q}$ is  \cite[Lemma IV.8]{ourJSAIT}.
The entropy bounds are via \cite[Lemma IV.7]{ourJSAIT} and (\ref{eq:lqgbound}) is \cite[Theorem IV.3 (iii)]{ourJSAIT}.
\end{proof} 

In the encoder/decoder in Fig. \ref{fig:overviewachiev}, at time $\t$ the quantization $\rvec{q}_{\t}$ is encoded into the prefix-free codeword $\rvec{a}_{\t}\in\{0,1\}^{*}$. The encoding is lossless, and upon receiving $\rvec{a}_{\t}$ the decoder can recover $\rvec{q}_{\t}$. The focus of the reminder of this work is on this lossless encoding. Let $C_{\t}:\mathbb{Z}^{\mathrm{m}}\rightarrow \{0,1\}^{*}$ be the encoding function used at time $\t$, e.g. assume $\rvec{a}_{\t}=C_{\t}(\rvec{q}_{\t})$. If the probability mass function  $\mathbb{P}_{\rvec{q}_{\t}}$ is known at every $\t$, $C_{\t}$ can be chosen as a SFE code adapted to the distribution $\mathbb{P}_{\rvec{q}_{\t}}$ and achieve a codeword length $\mathbb{E}[\ell(\rvec{a}_{\t})] \le H(\rvec{q}_{\t})+1$ (cf. e.g. \cite[Sec. IV.A.1]{ourJSAIT}). In \cite{ourJSAIT}, we proved that if the limiting distribution $\mathbb{P}_{\rvec{q}}$ is known and $C_{\t}$ is chosen as a fixed (time-invariant) SFE code adapted to the distribution $\mathbb{P}_{\rvec{q}}$, then 
${\lim\sup}_{\t\rightarrow \infty}\text{ } \mathbb{E}[\ell(\rvec{a}_{\t})] \le {\lim\sup}_{\t\rightarrow \infty} H(\rvec{q}_{\t})+1$. Thus, if \textit{either} the marginal PMFs $\mathbb{P}_{\rvec{q}_{\t}}$ or the limiting PMF  $\mathbb{P}_{\rvec{q}}$ is known, we can losslessly encode $\{\rvec{q}_{\t}\}$ so that the prefix constraint is satisfied and
\begin{IEEEeqnarray}{rCl}
  \underset{\s{T}\rightarrow \infty}  {\lim\sup}\text{ } \sum_{\s{i}=0}^{\s{T}-1}\mathbb{E}[\ell(\rvec{a}_{\s{i}})] &\le& \mathcal{R}(\gamma)+\s{b}+1,\label{eq:codewordlengthbounds}
\end{IEEEeqnarray} e.g. the system achieves a time-average codeword length that is at most $\s{b}$ bits above the fundamental rate-distortion lower bound in (\ref{eq:threestageRDF}). While the bound in (\ref{eq:codewordlengthbounds}) can be achieved if either the $\mathbb{P}_{\rvec{q}_{\t}}$  or   $\mathbb{P}_{\rvec{q}}$ are known, this is difficult to accomplish in practice. It is doubtful that they admit closed form characterizations. This motivates pursuing designs for $C_{\t}$ via universal coding. In our present context universal coding functions will be constructed via past observations of the source process. To make this explicit, we will denote $C_{\t} = C_{\t, \rvec{q}^{\t-1}}$ when $C_{\t}$ is a function of the ``past realizations" $\rvec{q}^{\t-1}$. We discuss universal coding in the next section. 

\section{Towards universal coding for the quantizer outputs}
We begin with some definitions. A probability mass function $\mathbb{P}_{\rs{z}}:\mathbb{N}_{+}\rightarrow [0,1]$ belongs to the \textit{power-law} envelope class with parameters $(\s{\alpha},\s{\beta})$ if for $\s{\alpha} >1$, $\beta > 2^{\alpha}/\left(\sum_{\s{i}=1}^{\infty}\frac{1}{\s{i}^{\s{\alpha}}}\right)$,     $\mathbb{P}_{\rs{z}}(\s{x}) \le \min\left(\frac{\beta
    }{\s{x}^{\alpha}},1\right)$ for $\s{x} \in\mathbb{N}_{+}$\cite{frenchOG}.
 Likewise, a probability mass function $\mathbb{P}_{\rs{z}}:\mathbb{N}_{+}\rightarrow [0,1]$ belongs to the exponential envelope class with parameters $(\s{\alpha},\s{\beta})$ if for  $\s{\alpha}>0$, $\s{\beta}>e^{2\alpha}$ such that  $\mathbb{P}_{\rs{z}}(\s{x}) \le \min(\beta e^{-\alpha \s{x} },1)$, for $\s{x} \in\mathbb{N}_{+}$. Finally, a (scalar) random variable $\rs{x}$ is $\s{\sigma}$-\textit{subgaussian} if $\mathbb{E}[\rs{x}]=0$ and $\mathbb{E}[e^{\s{\lambda} \rs{x}}] \le e^{\s{\lambda}^2\s{\sigma}^2/2}$ for all $\s{\lambda}\in\mathbb{R}$. A random vector  $\rvec{x}\in\mathbb{R}^{\s{m}}$ is $\s{\sigma}$-subgaussian if for every $c\in\mathbb{R}^{\s{m}}$ with $\lVert c\rVert_{2}=1$, $c^{\tp}\rvec{x}$ is $\s{\sigma}$-subgaussian. 

Let $\rs{x}_{\t}$ be a stationary, IID source on $\mathbb{N}_{+}$. Assume $\rs{x}_{\t}\sim \mathbb{P}_{\rs{x}}$ for all $\t$. Let $\rs{x}\sim\mathbb{P}_{\rs{x}}$. Assume that $\mathbb{P}_{\rs{x}}$ falls into either a power-law or exponential envelope class. A consequence of \cite{frenchOG} and \cite{frenchExp} is that one can construct a sequence of functions $C_{\t,\rs{x}^{\t-1}}:\mathbb{N}_{+}\rightarrow \{0,1\}^{*}$ such that 
\begin{align}\label{eq:zerodelay}
    \underset{\s{T}\rightarrow \infty}{\lim\sup}\frac{1}{\s{T}}\sum_{\s{i}=0}^{\s{T}-1}\mathbb{E}\left[\ell\left(C_{\t,\rs{x}^{\t-1}}\left(\rs{x}_{\t}\right)\right)\right] \le H(\rs{x})+1.
\end{align}Further, if $\mathbb{P}_{\rs{x}}[\rvec{x}_{\t}=x_{1}], \mathbb{P}_{\rs{x}}[\rvec{x}_{\t}=x_{2}]>0$ then $C_{\t}(x_{1})$ is not a prefix of $C_{\t}(x_{2})$ and vice-versa. This follows from replacing arithmetic coding with SFE coding (according to the same probability model) in the power-law (cf. \cite{frenchOG}) or exponential algorithms (cf. \cite{frenchExp}). The sequence $C_{\t}$ is constructed online  without explicit knowledge of the PMF $\mathbb{P}_{\rs{x}}$.  

The envelope classes, and adaptive encoding schemes in \cite{frenchOG,frenchExp,frenchSeq,frenchSeq2} apply to sources on $\mathbb{N}_{+}$, however, the quantizer output discussed in Section \ref{sec:sys} has $\rvec{q}_{\t}\in\mathbb{Z}^{\s{m}}$. We will develop a bijection $g:\mathbb{Z}^{\s{m}}\rightarrow\mathbb{N}_{+}$ 
such that the limiting distribution of $g(\rvec{q}_{\t})$ falls into
in an envelope class. In particular, assume that $g$ is such that
if $\lVert a \rVert_{\infty} > \lVert b \rVert_{\infty}$  then $g(a) > g(b)$ In other words, assume  $g$ is a function that respects the partial order induced by the infinity norm on $\mathbb{Z}^{\s{m}}$. Clearly such a mapping exists. For $\s{i}\in\mathbb{N}$, define $ \mathcal{B}_{\s{i}} = \{ z \in \mathbb{Z}^{\s{m}} : \lVert z \rVert_{\infty} \le \s{i} \},$ which is the set of points in  $\mathbb{Z}^{\s{m}}$ that lie within an origin-centered hypercube with edges of length $2\s{i}$. Note that the cardinality of $\mathcal{B}_{\s{i}}$ is  $|\mathcal{B}_{\s{i}}| = (2\s{i}+1)^{\s{m}}$. To define $g$, choose $g(0_{\s{m}}) = 1$, and then map the remaining points in  $\mathcal{B}_{1}$ arbitrarily to $2$ to $3^{\s{m}}$ and so on. We have the following theorem. 
\begin{theorem}\label{thm:pwrlaw}Let $g:\mathbb{Z}^{\s{m}}\rightarrow\mathbb{N}_{+}$ be any bijection such that
if $\lVert a \rVert_{\infty} > \lVert b \rVert_{\infty}$  then $g(a) > g(b)$. Let $\rvec{q}$ be as described in Proposition \ref{prop:prev}, and define $\overline{\rvec{q}} = g(\rvec{q})$. If $\s{m}=1$ or $\s{m}=2$, $\overline{\rvec{q}}$ falls into the exponential envelope class. If  $\s{m}>2$,  have that $\overline{\rvec{q}}$ is a member of a power-law envelope class. 
\end{theorem}
\begin{proof}  Let  $\{\rvec{\omega}_{\s{t}}\}$, $\{\rvec{\nu}_{\s{t}}\}$, and $\rvec{\chi}$ be as in Prop, \ref{prop:prev}, and let $\rvec{\delta}_{\t}\indep (\rvec{\chi},\{\rvec{\omega}_{\s{t}}\}, \{\rvec{\nu}_{\s{t}}\})$. Define $\rvec{z}_{\t}=C\rvec{e}_{\s{t}}+\rvec{\delta}_{\t}$. 
Via 
(\ref{eq:identicallydistributed}),  
\begin{align}\label{eq:rememberedef}
\rvec{z}_{\t}  \overset{\mathrm{D}}{=} C\left(R^{\s{t}}\rvec{\chi}+\sum_{\s{i}=0}^{\s{t}-1}R^{\s{i}}(\rvec{\omega}_{\s{i}}-L\rvec{\nu}_{\s{i}})\right)+\rvec{\delta}_{\t}.
\end{align} Take $c\in\mathbb{R}^{\s{m}}$ with $\lVert c \rVert_{2}=1$. Let $\kappa_{1} = \lim_{\t\rightarrow \infty }\sum_{\s{i}=0}^{\s{t}-1}\lVert R^{\s{i}}\rVert_{2}$, $\kappa_{2} = \sup_{\t} \lVert R^{\t}\rVert^{2}_{2}$, and $\kappa_{3} = \lim_{\t\rightarrow \infty }\sum_{\s{i}=0}^{\s{t}-1}\lVert R^{\s{i}}\rVert_{2}^{2}$. Note that since $\meig(R)< 1$, $\s{\kappa}_{1}$,  $\s{\kappa}_{2}$, $\s{\kappa}_{3}$ are finite via Gelfand's theorem (cf. e.g. Proposition A.4 in \cite{ourJSAIT}). Define $\underline{\rvec{z}}_{\t}= -C\sum_{\s{i}=0}^{\s{t}-1}R^{\s{i}}L\rvec{\nu}_{\s{i}}+\rvec{\delta}_{\t}$ and let $\overline{\rvec{z}}_{\t} = C (R^{\t} \rvec{\chi} + \sum_{\s{i}=0}^{\s{t}-1}R^{\s{i}}\rvec{\omega}_{\s{i}})$ so that  $\rvec{z}_{\t}\overset{\mathrm{D}}{=}\underline{\rvec{z}}_{\t}+\overline{\rvec{z}}_{\t}$ with $\mathbb{E}[\overline{\rvec{z}}_{\t}]=\mathbb{E}[\underline{\rvec{z}}_{\t}]=0$. For all $\t$, $c^{\tp}\underline{\rvec{z}}_{\t}$ has bounded support, i.e. 
\begin{align}\label{eq:boundedsupportpart}
    |c^{\tp}\underline{\rvec{z}}_{\t}|\le \left(\s{\kappa}_{1}\lVert C \rVert_{2}\lVert L \rVert_{2}+1\right) \frac{\sqrt{\s{m}}}{2}
\end{align} which follows via the triangle inequality, Cauchy-Schwartz, the submultiplicativity of matrix norms, and since $\lVert\rvec{\nu}_{\t}\rVert_{2}, \lVert\rvec{\delta}_{\t}\rVert_{2} \le \frac{\sqrt{\s{m}}}{2}$. Let $\s{\sigma}_{1} = \left(\s{\kappa}_{1}\lVert C \rVert_{2}\lVert L \rVert_{2}+1\right) \frac{\sqrt{\s{m}}}{2}$. Given (\ref{eq:boundedsupportpart}) we have that for every $\t$ and $c$ with $\lVert c\rVert_{2}=1$, $c^{\tp}\underline{\rvec{z}}_{\t}$ is $\s{\sigma}_{1}$-subgaussian \cite[Ex. 5.6 (b)]{ba}.  

Define   $\Omega_{\t} = C\left(R^{\s{t}}X_{0}(R^{\s{t}})^{\tp} + \sum_{\s{i}=0}^{\s{t}-1}R^{\s{i}}W(R^{\s{i}})^{\tp}\right)C^{\tp}$. We have $c^{\tp}\overline{\rvec{z}}_{\t}\sim\mathcal{N}(0,c^{\tp}\Omega_{\t}c)$. If we denote $\s{\sigma}^{2}_{2} =  \lVert C\rVert_{2}^{2}\left(\lVert X_{0}\rVert_{2}\s{\kappa}_{2}+\lVert W \rVert_{2}\s{\kappa}_{3}\right)$,
for all $\t$ we have $ c^{\tp}\Omega_{\t}c \le \s{\sigma}_{2}^{2}$,  which follows analagously by the triangle inequality, Cauchy-Schwartz, the submultiplicativity of matrix norms. Since $c^{\tp}\overline{\rvec{z}}_{\t}$ is a zero-mean Gaussian with a variance upper bounded by $\sigma^2_2$ we have that $c^{\tp}\overline{\rvec{z}}_{\t}$ is $\s{\sigma}_{2}$-subgaussian. Note that $\s{\sigma}_{2}$ does not depend on $c$, and that this holds for any $c$ with $\lVert c \rVert_{2}=1$.  Since $c^{\tp}\rvec{z}_{\t}\overset{\mathrm{D}}{=}c^{\tp}\underline{\rvec{z}}_{\t}+c^{\tp}\overline{\rvec{z}}_{\t}$, via \cite[Lemma 5.4 (b)]{ba} $c^{\tp}\rvec{z}_{\t}$ is $\sigma = \sqrt{\sigma_{1}^{2}+\sigma_{2}^{2}}$-subgaussian for all $\t$ and $c$ with $\lVert c\rVert_{2} = 1$. Thus the $\rvec{z}_{\t}$ is $\sigma$-subgaussian for all $\t$. 

Define $\rvec{q}_{t} = Q_{\Delta}(\rvec{z}_{t})$ and $\overline{\rvec{q}}_{\t} = g(\rvec{q}_{\t})$. Note that for $\s{r}\in\mathbb{N}_{0}$ by definition the set $\mathcal{B}_{\s{r}}$ contained exactly $|\mathcal{B}_{\s{r}}| = (2\s{r}+1)^{\s{m}}$ points. By definition of the bijection $g$,  for $\s{r}\in\mathbb{N}_{+}$, $\lVert \rvec{q}_{\t} \rVert_{\infty}\ge \s{r}$ if and only if $\overline{\rs{q}}_{\t} \ge (2(\s{r}-1)+1)^{\s{m}}+1$. 
Take $\s{z} \in \mathbb{N}_{+}$. We have 
\begin{IEEEeqnarray}{rCl}
    \mathbb{P}[\overline{\rs{q}}_{\t}=\s{z}] &\le& \mathbb{P}\left[ \overline{\rs{q}}_{t}\ge \s{z}\right] \label{eq:veryconservative}\\ &\le &  \mathbb{P}\left[\overline{\rs{q}}_{t}\ge \left(2(\lfloor \frac{\sqrt[\s{m}]{\s{z}-1}+1}{2}\rfloor -1)+1\right )^{\s{m}}+1\right] \nonumber  \\  &\le & \mathbb{P}\left[ \lVert \rvec{q}_{\t} \rVert_{\infty}\ge \lfloor \frac{\sqrt[\s{m}]{\s{z}-1}+1}{2}\rfloor\right]\label{eq:explainplease},
\end{IEEEeqnarray}  where the $\s{z}=1$ case in (\ref{eq:explainplease}) holds trivially and for $\s{z}>1$ (\ref{eq:explainplease}) follows from our observation above (\ref{eq:veryconservative}).   Note that for $\s{r}\in\mathbb{N}_{0}$, $\lVert \rvec{q}_{\t} \rVert_{\infty}\ge \s{r} $ if and only if $\lVert C\rvec{e}_{\t}+\rvec{\delta}_{\t} \rVert_{\infty} \ge \s{r}-\frac{1}{2}$. Thus, further relaxing (\ref{eq:explainplease})
\begin{IEEEeqnarray}{rCl}
    \mathbb{P}[\overline{\rs{q}}_{\t}=\s{z}] &\le& \mathbb{P}\left[  \lVert C\rvec{e}_{\t}+\rvec{\delta}_{\t} \rVert_{\infty}\ge    \lfloor \frac{\sqrt[\s{m}]{\s{z}-1}+1}{2}\rfloor-\frac{1}{2} \right]\nonumber\\ &\le& \mathbb{P}\left[  \lVert C\rvec{e}_{\t}+\rvec{\delta}_{\t} \rVert_{\infty}\ge      \left(\frac{\sqrt[\s{m}]{\s{z}-1}}{2}-1\right)\right]\\ &\le& 2\s{m} e^{-\frac{1}{2\sigma^2} \left( \frac{1}{2}\sqrt[\s{m}]{\s{z}-1}-1\right)^{2}} \label{eq:usetoulouse}
\end{IEEEeqnarray} where (\ref{eq:usetoulouse}) follows from a maximal inequality for subgaussian random vectors (cf. \cite[Theorem 2.2.1]{frenchhds}). Let $\zeta = 1/(8\sigma^2)$.  We have 
$\s{z}^{2} \le e^{\zeta \left(\sqrt[\s{m}]{\s{z}-1}-2\right)^{2}}$ for $\s{z}$ sufficiently large, and thus, $\mathbb{P}[\overline{\rs{q}}_{\t}=\s{z}] \le \frac{2\s{m}}{\s{z}^{2}}$  for $\s{z}$ sufficiently large. This demonstrates that irrespective of $\s{m}$, the quantizer output, when wrapped by the function $g$, falls into a power-law envelope class with parameter $\alpha = 2$ and a $\beta$ that depends on $\zeta$ and $\s{m}$.  If $\s{m}=1$, once we have immediately that $\overline{\rs{q}}_{\t}$ is in the exponential class with $\alpha = \zeta$ and a $\beta$ that depends on $\zeta$. If $\s{m}=2$, $\overline{\rs{q}}_{\t}$ is in the exponential class with $\alpha = \frac{\zeta}{2}$ and a $\beta$ that depends on $\zeta$. Since the bound in  (\ref{eq:usetoulouse}) holds for all $\t$, it holds for the limiting distribution $\overline{\rs{q}}=g(\rvec{q})$.
\end{proof}

Note that $H(\rvec{q}_{\t}) = H(\overline{\rs{q}}_{\t})$ and  $H(\rvec{q}) = H(\overline{\rs{q}})$ since $g$ is a bijection. Likewise, by Proposition \ref{prop:prev} $D_{\mathrm{KL}}(\overline{\rs{q}}_{\t}||\overline{\rs{q}})\rightarrow 0$. A corollary to this latter fact, and Theorem \ref{thm:pwrlaw}'s claim that $\overline{\rs{q}}$ falls into a power-law or exponential envelope class is that $\lim_{\t\rightarrow \infty}H(\rvec{q}_{\t}) = H(\rvec{q})$ \cite[Theorem 21]{infotops}. The results in \cite{frenchOG} and \cite{frenchExp} hold only for \textit{stationary sources} on $\mathbb{N}_{+}$ meanwhile $\overline{\rs{q}}_{\t}$ is only asymptotically stationary. However, given Proposition \ref{prop:prev} and Theorem \ref{thm:pwrlaw} we conjecture the following. 

 \begin{conj}\label{conj:possiblyembarasing}
  Let $g$ be a bijection as described in Theorem \ref{thm:pwrlaw}. Assume that the lossless encoder in Fig. \ref{fig:overviewachiev} first computes  $\overline{\rs{q}}_{\t} = g(\rvec{q}_{\t})$, and then encodes $\overline{\rs{q}}_{\t}$ via $\rvec{a}_{\t}=C_{\t,\overline{\rs{q}}^{\t-1}}(\overline{\rs{q}}_{\t})$, where $C_{\t,\overline{\rs{q}}^{\t-1}}$ is the encoding function constructed via either \cite{frenchExp} (in the case that $\s{m}\le 2$) or \cite{frenchOG} with the zero-delay modification (replacing arithmetic coding with SFE coding, cf. (\ref{eq:zerodelay})). The encoding is prefix-free in the sense of Section \ref{sec:sys} (indeed, particular prefix-free encoding used at time $\t$ depends only on $\overline{\rs{q}}^{\t-1}$), and the decoder can reconstruct $\rvec{q}_{\t}$ exactly, thus ensuring that the constraint on LQG control performance is satisfied. We conjecture that the time-average expected codeword lengths will satisfy $\lim_{\s{T}\rightarrow \infty}\frac{1}{\s{T}}\sum_{\s{i}=0}^{\s{T}-1}\mathbb{E}[\ell(\rvec{a}_{\t})] = H(\rvec{q})+2$.
 \end{conj} 
 
 While proving (or disproving) Conjecture \ref{conj:possiblyembarasing} is a topic for future work; we suspect that it holds  given the ergodicity of $\{\rvec{q}_{\t}\}$ and since $\rvec{q}_{\t}\rightarrow\rvec{q}$ in the KL-sense. 
 
\section{A practical adaptive algorithm}\label{sec:alg} A key element of the coding schemes in \cite{frenchOG,frenchExp,frenchSeq,frenchSeq2} was the notion of cutoffs. Generally speaking, if the source produced a symbol below the cuttoff at at a given time, the symbol was encoded via arithmetic coding. If the source produced a symbol that exceeded the cutoff, en escape symbol was encoded arithmetically, followed by the symbol encoded using a fixed (Elias) universal code. In \cite{frenchOG,frenchExp,frenchSeq,frenchSeq2}, the cuttoffs grow with time, and (in our case) will diverge as large numbers of symbols are encoded. Cutoffs that grow with time are impractical as they require that the arithmetic precision used to implement arithmetic, or, in the zero-delay case, SFE, encoders and decoders similarly expands over time \cite{wittenArithmetic}. Since our interest is in long term, infinite horizon communication cost, we propose instead to fix the cutoffs a priori and account for fixed arithmetic precision. We describe the algorithm in the remainder of this section. 

We first describe our notion of ``censoring" before describing the encoding. Let $\s{p}$ be the precision, in bits, in which (unsigned integer) arithmetic operations are to be performed at the encoder and decoder. We first transform $\rvec{q}_{\t}$ into a source on $\mathbb{N}_{+}^{\s{m}}$ by computing the vector $\rvec{s}_{\t}$ with elements
\begin{align}\label{eq:wrap}
    [\rvec{s}_{\t}]_{\s{i}} = \begin{cases}
         2[\rvec{q}_{\t}]_{\s{i}}\text{, }&[\rvec{q}_{\t}]_{\s{i}}>0 \\ 
          -2[\rvec{q}_{\t}]_{\s{i}}+1\text{, }&[\rvec{q}_{\t}]_{\s{i}}\le 0 \\ 
    \end{cases}.
\end{align} Let ${k}\in\mathbb{N}^{\s{m}}_{+}$ be a vector of \textit{cutoffs}. The ${k}$ are fixed a priori; and are hyperparameters of our encoding algorithm. We require that $ \s{n}=\prod_{\s{j}=1}^{\s{m}}([k]_{\s{j}}+1)$ has $\s{n} < 2^{\frac{\s{p}}{2}}.$ Define the truncation operator $\text{trunk}_{k}:\mathbb{N}_{+}^{\s{m}}\rightarrow\mathbb{N}_{0}^{\s{m}}$  via
\begin{align}
   [\text{trunk}_{k}(s)]_{\s{i}}  = \begin{cases}
         [s]_{\s{i}}\text{, }&  [s]_{\s{i}} \le [k]_{\s{i}} \\ 
          0\text{, }&\text{otherwise}\\ 
    \end{cases}.
\end{align} Define the post-truncation symbol tuple $\overline{\rvec{s}}_{\t} = \text{trunk}_{k}(\rvec{s}_{\t})$, which is a source on an alphabet of cardinality $\s{n}$.  Denote the sequence of symbols that were truncated $\hat{\rvec{s}}_{\s{t}}$, so that $[\hat{\rvec{s}}_{\t}]_{\s{i}}$ is the  $\s{i}^{\mathrm{th}}$ symbol truncated from $\rvec{s}_{\t}$. Note that the dimension of the vector  $\hat{\rvec{s}}_{\s{i}}$  is random.  Define the \textit{linear indexing} function $\lambda:\mathbb{N}_{0}^{\s{m}}\rightarrow \mathbb{N}_{0}$  via $\lambda(s) = \sum_{\s{i}=1}^{\s{m}-1}  ([s]_{\s{i}})\prod_{\s{j}=\s{i}+1}^{\s{m}}([k]_{\s{j}}+1)+[s]_{\s{m}}$. We have that $\lambda$ is a bijection from the range of $\rvec{s}_{\t}$  to the set $\{0, 1, \dots, \s{n}-1\}$. Let  $\rs{m}_{\s{t}} =   \lambda(\overline{\rvec{s}}_{\t})$. In our algorithm, we encode $\rs{m}_{\s{t}}$ using a SFE code (cf. \cite{coverandthomas}, \cite[Section IV.A.1]{ourJSAIT}). Subsequently, we encode each element of $\hat{\rvec{s}}_{\s{t}}$ with the Elias omega code \cite{eliasUniversal}. The decoder decodes $\rs{m}_{\s{t}}$, and recovers $\overline{\rvec{s}}_{\t}$ via inverting the linear indexing function $\lambda$. The decoder counts the overflow ``$0$" symbols in $\overline{\rvec{s}}_{\t}$, and decodes the omega-encoded $\hat{\rvec{s}}_{\s{t}}$ from the remaining bits. The decoder then reconstructs $\rvec{s}_{\t}$ and $\rvec{q}_{\t}$. 

We implement the SFE encoding along the same lines as the fixed-precision implementation of adaptive arithmetic coding in
\cite{wittenArithmetic}. To encode $\rs{m}_{\s{t}}$  via an SFE codec, a probability mass function for $\rs{m}_{\s{t}}$ is required. We encode $\rs{m}_{\s{t}}$ using an empirical model based on $\rs{m}^{\s{t}-1}$. Let $\rvec{c}_{\t}\in\mathbb{N}_{0}^{\s{n}}$, and assume $[\rvec{c}_{-1}]_{\s{i}}=1$ for $\s{i} \in \{0, \dots, \s{n}-1\}$. For all $\t$, let $\rs{r}_{\t}=\sum_{\s{i}=0}^{\s{n}}  [\rvec{c}_{\t}]_{\s{i}}$. The PMF used for SFE encoding at time $\s{t}+1$ is based on the empirical frequencies
\begin{align}\nonumber
    [\rvec{c}_{\t}]_{\s{i}} = \begin{cases}
        \frac{1}{2}\left([\rvec{c}_{\t-1}]_{\s{i}}-1\right)+1+\mathbb{1}_{\rs{m}_{\s{t}}=\s{i}},\text{ if }\rs{r}_{\t-1} =  2^{\s{p}/2}-1\\ 
  [\rvec{c}_{\t-1}]_{\s{i}}+\mathbb{1}_{\rs{m}_{\s{t}}=\s{i}},\text{ otherwise}.
    \end{cases}
\end{align} We  encode $\rvec{m}_{\t}$ using the PMF $\mathbb{P}_{\rvec{c}_{\t-1}}(\s{i}) = [\rvec{c}_{\s{t}-1}]_{\s{i}}$. Both the encoder and decoder begin with the same initial model $\rvec{c}_{-1}$, and the encoder and decoder updates their models via  after encoding/decoding so that they remain synchronized. The update rule periodically re-scales to ensures that the arithmetic operations required for SFE encoding/decoding can be carried out in $\s{p}$ bits of precision. For more details, see \cite[Chapter 11]{cuvelierDissertation}.
Since the model used for SFE encoding at time $\t$ depends only on $\rs{m}^{\s{t}-1}$, the SFE coding is prefix-free in the sense of Section \ref{sec:sys} (given $\rs{m}^{\s{t}-1}$). The Elias omega coding is likewise instantaneously decodable, and thus the jointly encoding $\rvec{q}_{\t}$ via SFE encoding   $\overline{\rvec{s}}_{\t}$ and $\overline{\rvec{s}}_{\t}$ satisfies our desired prefix constraint. We now characterize this algorithm's performance numerically.  
\section{Numerical results}
We perform our experiments using a linearized model for the inverted pendulum system from \cite{ssmodelinvertedpendulum}. We describe the inverted pendulum in detail in appendix Appendix \ref{apx:dynamics}. The system's state vector consists of $\s{m}=4$ dimensions including a horizontal position and velocity and the pendulum's azimuthal angle, and angular velocity,  with respect to the normal from the cart's platform. The system's control input is $\s{u}=1$ dimensional.  We discretize the system's dynamics at a sampling frequency of 100 samples per second and assuming that the continuous-time control input is via sample-and-hold at the same frequency. The discrete-time system dynamics are assumed to be
$\rvec{x}_{\t+1} = A_{\tau} \rvec{x}_{\t} + B_{\tau }\rvec{u}_{\t}+ \rvec{{w}}_{t}$ where  $A_{\tau}\in\mathbb{R}^{4\times 4}$, $B_{\tau}\in\mathbb{R}^{4\times 1}$, and
$\rvec{w}_{\t}$ is IID process noise such that $\rvec{w}_{\t}\indep \rvec{u}_{0}^{\t-1}, \rvec{x}_{0}^{\t}, \rvec{w}_{0}^{\t-1}$ and $\rvec{w}_{\t}\sim \mathcal{N}(0,W_{\tau})$.  We assume a diagonal  $W_{\tau}$ such that $W_{\tau} =  .005I_{4}$, and an uncertain initial configuration such that $\rvec{x}_{0} \sim \mathcal{N}(0,.05 I_{4})$. We will assume that the LQG cost weights are given by $Q = I_{4\times 4}$ and $\Phi = 1$. 

All simulations were performed using MATLAB R2022B \cite{matlab}. We implemented the quantizer design from Section \ref{sec:sys} and encoded the quantizations using the compression algorithm of Section \ref{sec:alg}. To obtain the solution to the rate-distortion optimization in (\ref{eq:threestageRDF}), and to obtain the measurement matrix $C$, we used the YALMIP toolbox \cite{yalmip} with the MOSEK solver\cite{mosek}. We used the ``sorted" version of SFE coding in our implementation of the lossless compresssion algorithm; this tends to reduce redundancy (cf. \cite[Section IV.A.1]{ourJSAIT}). The SFE coding was implemented using unsigned 64-bit arithmetic, and was based on the arithmetic coding implementation in \cite{wittenArithmetic}. After some initial tuning, settled on a cutoff vector of $k = [8191,3,3,3]^{\tp}$. More details can be found in \cite[Chapter 11]{cuvelierDissertation}. 

 In Figure \ref{fig:empirical} we plot the empirical attained cost for quantizer designs obtained by varying the target LQG cost $\s{\gamma}$. We also plot the lower bound, $\mathcal{R}(\gamma)$, from (\ref{eq:threestageRDF}) and the upper bound from \cite[Theorem IV.3 (ii) (22)]{ourJSAIT}, which is approximately $\mathcal{R}(\gamma)+1.26\s{m}+1$ bits. The behavior of the rate-distortion lower bound is typical; the horizontal asymptote of the lower bound corresponds to the minimum feasible LQG performance with fully observable state feedback (at this sample rate). The lower bound on data rate rises sharply near this asymptotic and becomes quite modest as a higher cost can be tolerated.  We simulated the system for $\s{T}=400000$ samples, and plotted the average control cost versus the average codeword length over the entire horizon. We found that the control cost is rather slow to settle. The target control costs for most, but not all, of the points displayed are below the computed empirical average cost. Figure \ref{fig:joint} illustrates the convergence of the running average bitrate and control cost. 
\begin{figure}
	\centering
	\includegraphics[width = .9\columnwidth]{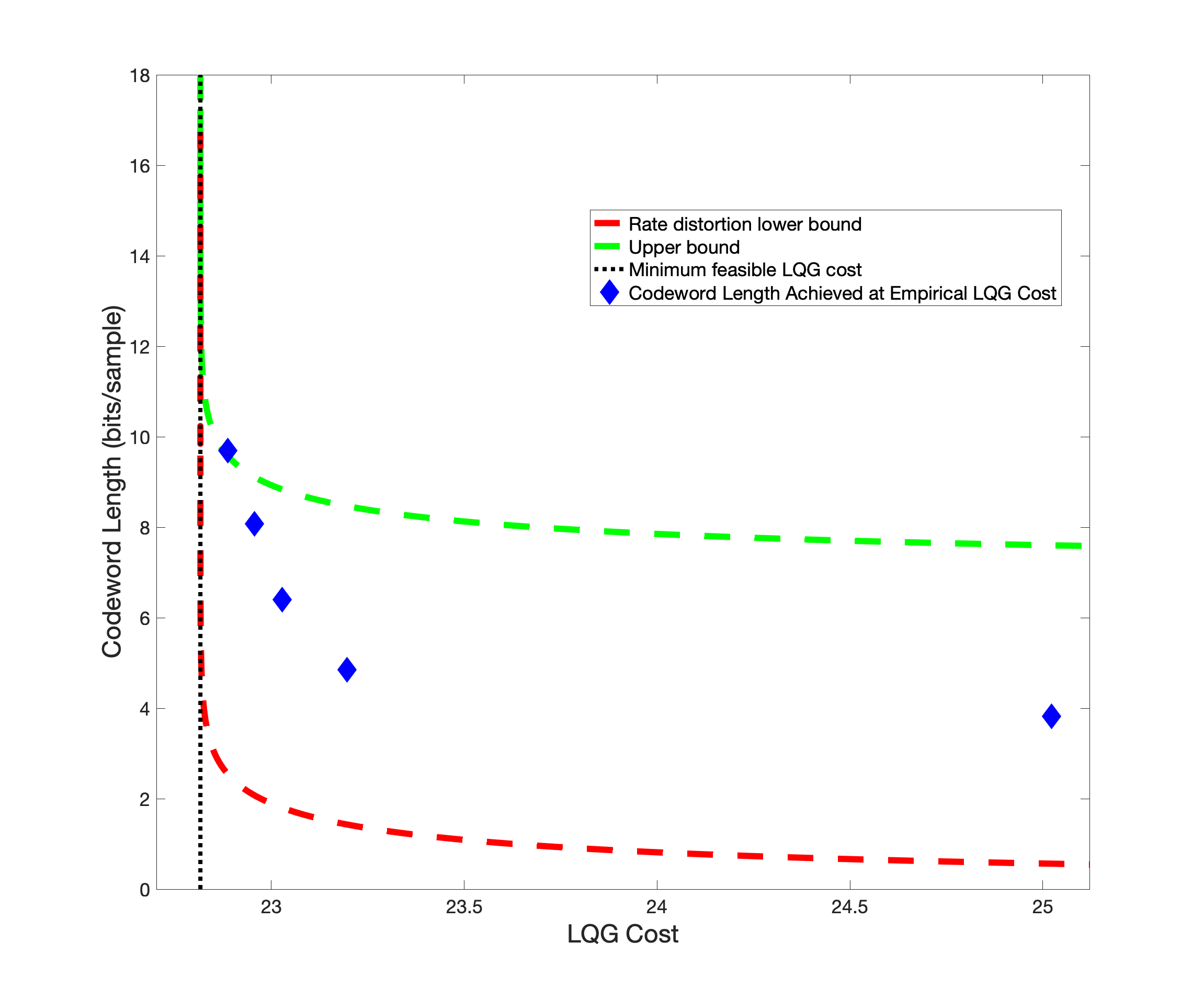}
   \vspace{-.5cm}
	\caption{ The control cost constraint, $\gamma$ is plotted on the horizontal, while the data rate in bits is plotted along the vertical. }\label{fig:empirical}
\end{figure} 
\begin{figure}
	\centering
	\includegraphics[width =.9\columnwidth]{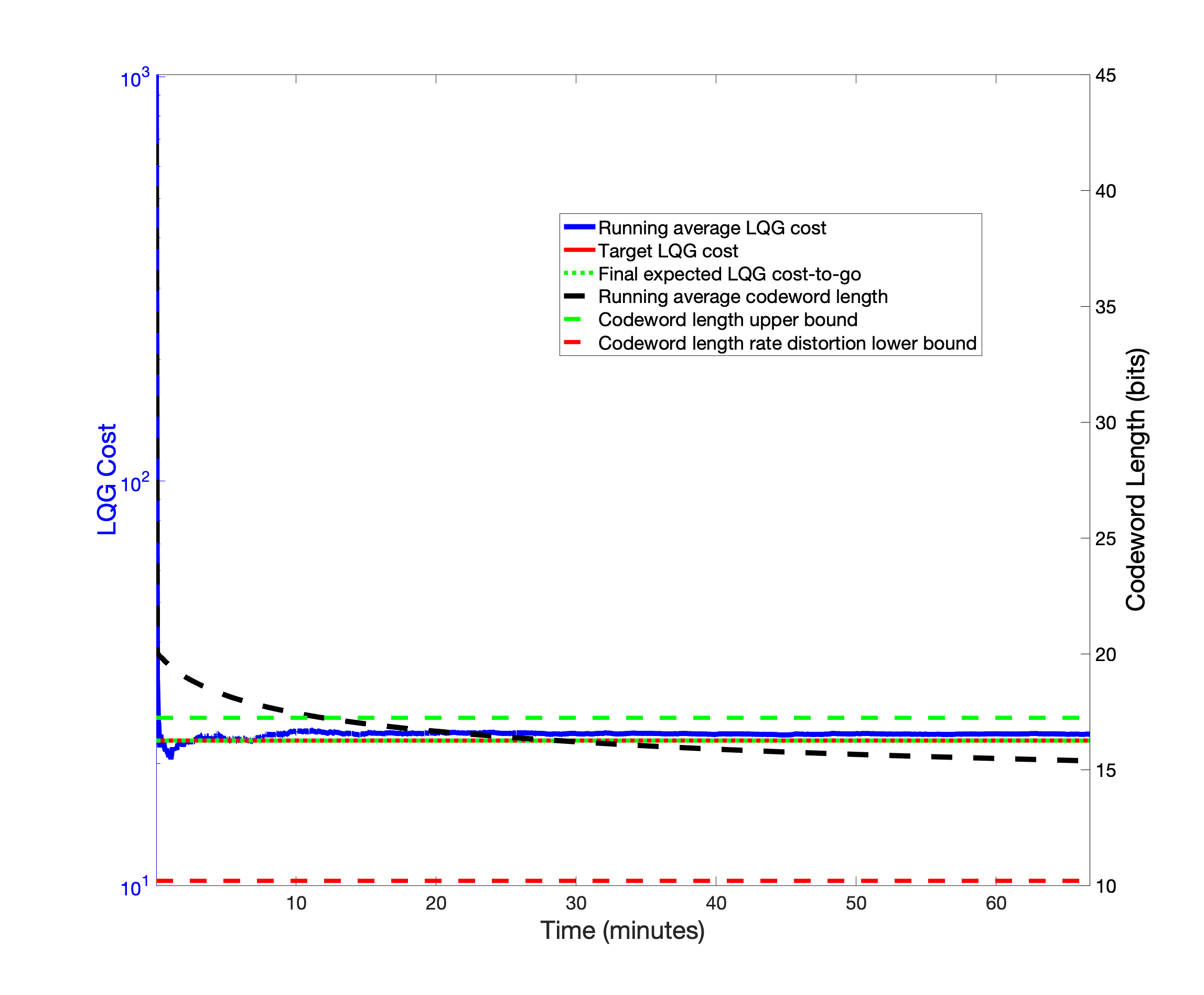}
  \vspace{-.5cm}
	\caption{The running average control and communication costs on the same plot for a particular fixed target control cost. We depict upper and lower bounds on codeword length, the target control cost, as well expected cost to go after $\s{T}= 400000$ samples. The expected cost to go is seen to converge to the target control cost. }\label{fig:joint}
\end{figure} 
\section{Conclusions} 
In interpreting these numerical results, one must take care to not interpret them as a ``real-world" experiment. They are performed using pseudorandom numbers and finite precision arithmetic using a standard top-of-the-line consumer laptop.  Such simulations doubtlessly suffer from numerical inaccuracies that may be significant in some applications.

Instead of encoding the linear index of the truncated symbol $\text{trunk}_{k}(\rvec{s}_{\t})$ with a (sorted) SFE code, one could encode $\text{trunk}_{k}(\rvec{s}_{\t})$ using arithmetic coding over the dimensions of the vector; viewing the elements of the vector as a Markov source. For example, the  encoder and decoder could store one model for the first component $[\text{trunk}_{k}(\rvec{s}_{\t})]_{0}$ of the source, then   $[k]_{0}+1$ models for $[\text{trunk}_{k}(\rvec{s}_{\t})]_{1}$,  one for each potential realization of $[\text{trunk}_{k}(\rvec{s}_{\t})]_{0}$, and so on. This sort of approach allows a more accurate characterization of the source probability mass function at the expense of greater spatial complexity. This also reduces the frequency of model rescalings.

\appendices 
\section{The inverted pendulum's system dynamics}\label{apx:dynamics}
\begin{figure}
	\centering
	\includegraphics[scale = .18]{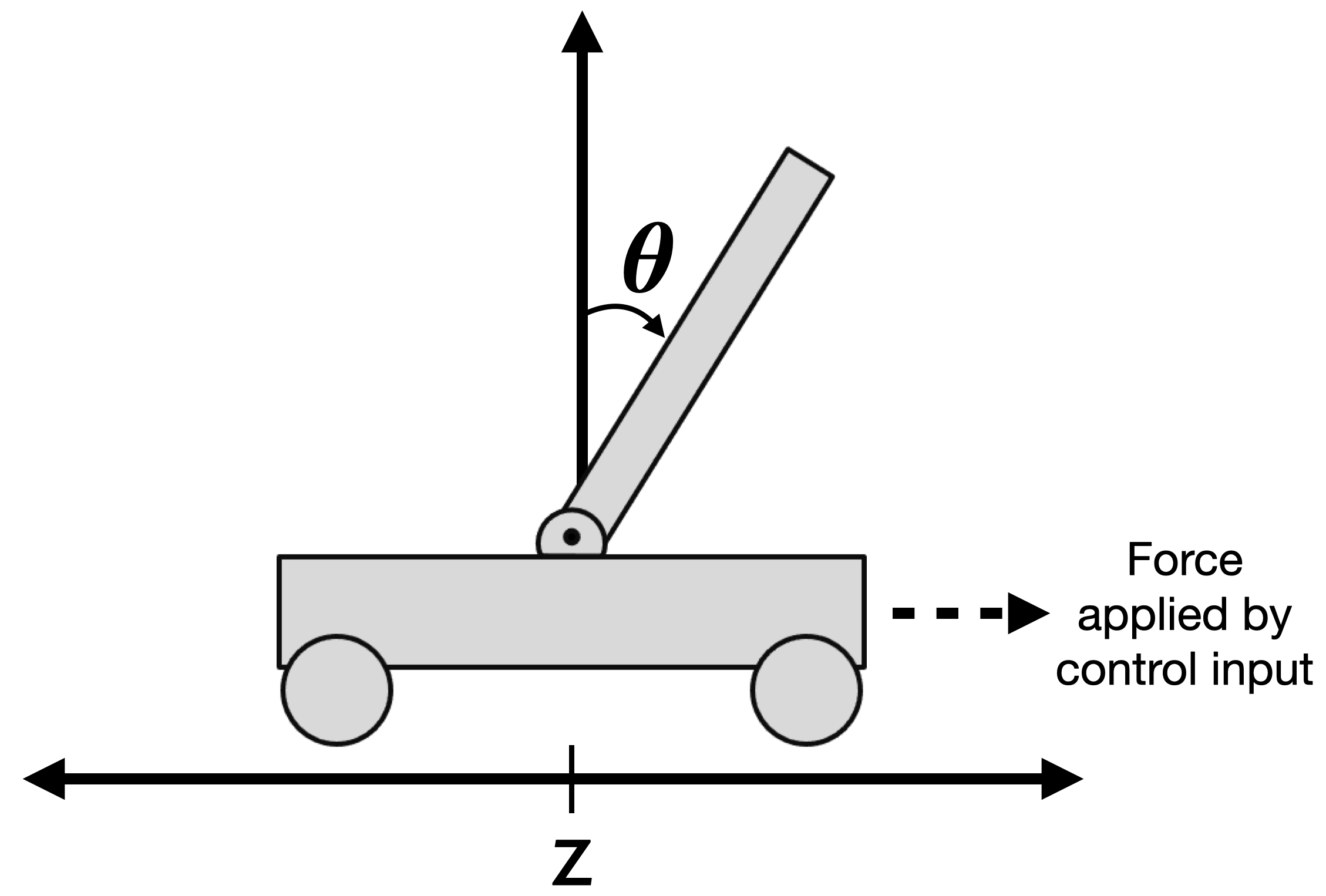}
	\caption{ The inverted pendulum, or ``cart pole" system consists of a motorized wheeled cart that can move in one dimension along the $z$ axis. The ``inverted pendulum" is affixed to the top of the cart, and consists of a slender armature anchored to a fulcrum on the cart. The control input can accelerate the cart along its axis of motion in an effort to stabilize the pendulum about its unstable equilibrium at $\theta = 0$.  }\label{fig:cartpole}
\end{figure}
\begin{table}
\centering
\resizebox{.9\columnwidth}{!}{
\begin{tabular}{ccllc}
\hline
\multicolumn{1}{|c|}{Variable} & \multicolumn{2}{|c|}{Descriptions}     & \multicolumn{2}{c|}{Value}   \\ \hline\hline
\multicolumn{1}{|c|}{$\s{\mu}_{\text{cart}}$}                     & \multicolumn{2}{|c|}{mass of cart}      & \multicolumn{2}{|c|}{.5 kilograms}     \\ \hline \multicolumn{1}{|c|}{$\s{\mu}_{\text{pend}}$}                     & \multicolumn{2}{|c|}{mass of pendulum}      & \multicolumn{2}{|c|}{.2 kilograms}     \\ \hline \multicolumn{1}{|c|}{$\s{\kappa}$}                     & \multicolumn{2}{|c|}{coefficient of friction for cart}      & \multicolumn{2}{|c|}{.1 newton/(meters  sec)} \\ \hline \multicolumn{1}{|c|}{$\s{\psi}$} & \multicolumn{2}{|c|}{mass moment of inertia for pendulum}      & \multicolumn{2}{|c|}{.006 kilogram meters$^{2}$} 
\\ \hline \multicolumn{1}{|c|}{$\s{\epsilon}$} & \multicolumn{2}{|c|}{length of pendulum to center of mass}      & \multicolumn{2}{|c|}{.3  meters}   \\ \hline \multicolumn{1}{|c|}{$\s{g}$} & \multicolumn{2}{|c|}{gravitational acceleration}      & \multicolumn{2}{|c|}{9.8  meters/(sec)$^2$}   \\ \hline
\end{tabular}}
\caption{Parameters of the inverted pendulum system in Fig. \ref{fig:cartpole} \cite{ssmodelinvertedpendulum}.}\label{tab:cartpoleparams}
\vspace{-.9cm}

\end{table}
We perform our experiments using a linearized model for the inverted pendulum system from \cite{ssmodelinvertedpendulum}. The inverted pendulum system is depicted and described in Fig. \ref{fig:cartpole}. The systems state is considered in $\s{m}=4$ dimensions, with $\rs{z}$ the lateral position of the cart along its axis of motion (in meters), $\dot{\rs{z}}=\frac{d}{d\t}\rs{z}$ the associated velocity (in meters/second), $\rs{\theta}$ the angle of the pole from the vertical (in radians), and $\dot{\rs{\theta}}=\frac{d}{d\t}\rs{\theta}$ the associated angular velocity (in radians/sec). The control input $\rvec{u}$ is one dimensional. For $\t\in\mathbb{R}$, $\t\ge 0$, denote the continuous time state vector
\begin{align}
    \rvec{x}(\t) = \begin{bmatrix}
     \rs{z}(\t) \\ \dot{\rs{z}}(\t) \\ \rs{\theta}(\t) \\ \dot{\rs{\theta}}(\t)
    \end{bmatrix}. 
\end{align} The true dynamics of the system are nonlinear, however, after linearizing, the continuous-time are assumed to be \cite{ssmodelinvertedpendulum}
\begin{align}\label{eq:ctmodel}
     d \rvec{x}(\t) = A_{\mathrm{ct}} \rvec{x}(\t) +B_{\mathrm{ct}} \rvec{u}(\t)+  W^{\frac{1}{2}}_{\mathrm{ct}} d\rvec{w}(\t),
\end{align} where the system matrix  $A_{\mathrm{ct}}$ and feedback matrix $B_{\mathrm{ct}} $ are functions of the system parameters given in Table \ref{tab:cartpoleparams}. Let 
\begin{align}
    \s{\rho} = \psi(\s{\mu}_{\text{pend}}+\s{\mu}_{\text{cart}})+\s{\mu}_{\text{pend}}\s{\mu}_{\text{cart}}\s{\epsilon}^{2}
\end{align}
Explicitly, we have 
\begin{align}
    A_{\mathrm{ct}}  =  \begin{bmatrix}
     0 & 1 & 0 & 0 \\ 0 & -\dfrac{ \left(\psi + \s{\mu}_{\text{pend}}\s{\epsilon}^{2}\right)\s{\kappa}   }{\s{\rho}} 
    & \dfrac{ \left( \s{\mu}^{2}_{\text{pend}}\s{\epsilon}^{2} \s{g}\right)  }{\s{\rho}}  &0 \\ 0 & 0 & 0 & 1\\ 0 & -\dfrac{\left(\s{\mu}_{\text{pend}}\epsilon \s{\kappa} \right)}{\s{\rho}} & \dfrac{\s{\mu}_{\text{pend}}\s{g}\s{\epsilon}(\s{\mu}_{\text{cart}}+\s{\mu}_{\text{pend}})}{\s{\rho}}& 0 
    \end{bmatrix},
\end{align}
\begin{align}
    B_{\mathrm{ct}}  =  \begin{bmatrix}
    0 \\ \dfrac{\left(\psi + \s{\mu}_{\text{pend}} \s{\epsilon}^{2}\right)}{\rho} \\ 0 \\ \dfrac{\s{\mu}_{\text{pend}}\s{\epsilon}}{\s{\rho}}
    \end{bmatrix},
\end{align} and $d\rvec{w}(\t)$ is  standard Brownian motion that accounts for modeling error and unmodeled dynamics (wind, etc). We now construct a discrete-time version of the model in (\ref{eq:ctmodel}), assuming a sampling period $\s{\tau}$.  We assume a sample-and-hold feedback policy where $u(\t) = u_{\lfloor\frac{\t}{\tau}\rfloor}$.
Let
\begin{subequations}\label{eq:dtssmodel}
\begin{align}
{A}_{\s{\tau}}=e^{{A}_{\mathrm{ct}}\s{\tau}},
\end{align}
\begin{align}
{B}_{\tau}=\int_0^{\tau} e^{{A}_{\mathrm{ct}} \s{s}}B_{\mathrm{ct}}d\s{s},
\end{align} and 
\begin{align}
    {W}_{\tau} = \left(\int_0^{\s{\tau}} e^{{A}\s{s}} {W}_{\mathrm{ct}}  e^{ A_{\mathrm{ct}}^\tp\s{s}}d\s{s}\right)^{\frac{1}{2}},
    \end{align}
\end{subequations} where for a matrix $M$ argument, $e^{M}$ refers to the matrix exponential of $M$.  For $\t \in \mathbb{N}_{0}$, we let $\rvec{x}(\t\tau) = \rvec{x}_{\t}$ and $\rvec{u}(\t\tau) = \rvec{x}_{\t}$.

\bibliographystyle{IEEEtran}
\bibliography{refs.bib}

\end{document}